\RequirePackage{ifpdf}
\ifpdf % We are running pdfTeX in pdf mode
\documentclass[pdftex]{sigma}
\else
\documentclass{sigma}
\fi

\begin{document}

\allowdisplaybreaks

\renewcommand{\PaperNumber}{011}

\FirstPageHeading

\renewcommand{\thefootnote}{$\star$}

\ShortArticleName{Classical $R$-Operators and Integrable
Generalizations  of Thirring Equations}

\ArticleName{Classical $\boldsymbol{R}$-Operators and Integrable
Generalizations\\ of Thirring Equations\footnote{This paper is a
contribution to the Proceedings of the Seventh International
Conference ``Symmetry in Nonlinear Mathematical Physics'' (June
24--30, 2007, Kyiv, Ukraine). The full collection is available at
\href{http://www.emis.de/journals/SIGMA/symmetry2007.html}{http://www.emis.de/journals/SIGMA/symmetry2007.html}}}

\Author{Taras V. SKRYPNYK~$^{\dag\ddag}$}

\AuthorNameForHeading{T.V. Skrypnyk}

\Address{$^\dag$~SISSA, via Beirut 2-4, 34014 Trieste, Italy}

\EmailD{\href{mailto:skrypnyk@sissa.it}{skrypnyk@sissa.it}}

\Address{$^\dag$~Bogolyubov Institute for Theoretical Physics,
       14-b Metrologichna Str., Kyiv 03680, Ukraine}

\ArticleDates{Received October 31, 2007, in f\/inal form January
18, 2008; Published online February 01, 2008}

\Abstract{We construct dif\/ferent integrable  generalizations of
the massive Thirring equations corresponding  loop algebras
$\widetilde{\mathfrak{g}}^{\sigma}$ in dif\/ferent gradings and
associated ``triangular'' $R$-operators. We consider the most
interesting cases connected with the Coxeter automorphisms, second
order automorphisms and with ``Kostant--Adler--Symes''
$R$-operators. We recover a known matrix generalization of the
complex Thirring equations as a partial case of our construction.}


\Classification{17B70; 17B80; 37K10; 37K30; 70H06}

\section{Introduction}
A theory of  hierarchies of integrable equations in partial
derivatives is based on the possibility to represent each of the
equations of the hierarchy in the so-called zero-curvature form:
\begin{equation}\label{zce}
\frac{\partial U(x,t,u)}{\partial t} - \frac{\partial
V(x,t,u)}{\partial x} + [U(x,t,u), V(x,t,u)]= 0,
\end{equation}
where $U(x,t,u)$, $V(x,t,u)$ are $\mathfrak{g}$-valued functions
with  dynamical variable coef\/f\/icients, $\mathfrak{g}$ is
simple (reductive) Lie algebra and $u$ is an additional complex
parameter usually called spectral. In order for the equation
(\ref{zce}) to be consistent it is necessary that $U(x,t,u)$ and
$V(x,t,u)$ belong to some closed inf\/inite-dimensional Lie
algebra $\widetilde{\mathfrak{g}}$ of  $\mathfrak{g}$-valued
functions of~$u$.

There are several approaches to  construction of zero-curvature
equations (\ref{zce}) starting from Lie algebras
$\widetilde{\mathfrak{g}}$. All of them are based on the so-called
Kostant--Adler--Symes scheme. One of the most simple and general
approaches is the approach of \cite{FNR,New} and \cite{Hol1,Hol2}
that interprets equation (\ref{zce}) as a consistency condition
for a two commuting Hamiltonian f\/lows written in the
Euler--Arnold or Lax form. In the framework of this approach
elements $U(x,t,u)$ and $V(x,t,u)$ coincide with the
algebra-valued gradients of the commuting Hamiltonians constructed
with the help of the
 Kostant--Adler--Symes scheme, the cornerstone of which is a
decomposition of the Lie algebra $\widetilde{\mathfrak{g}}$ (as a
vector space) in a direct sum of its two subalgebras:
\[
\widetilde{\mathfrak{g}}= \widetilde{\mathfrak{g}}_+ +
\widetilde{\mathfrak{g}}_- .
\] The
 algebra-valued gradients of the commuting Hamiltonians coincide with the
 restrictions of the algebra-valued gradients of Casimir
 functions of $\widetilde{\mathfrak{g}}$ onto the subalgebras
 $\widetilde{\mathfrak{g}}_{\pm}$.
Hence, such approach permits \cite{Skr8} to construct using Lie
algebra $\widetilde{\mathfrak{g}}$ the three types of integrable
equations: two types of equations with $U$-$V$ pair belonging to
the same Lie subalgebras $\widetilde{\mathfrak{g}}_{\pm}$ and the
third type of equations with $U$-operator belonging to
$\widetilde{\mathfrak{g}}_{+}$ and $V$-operator belonging to
$\widetilde{\mathfrak{g}}_{-}$ (or vise verse). The latter
equations are sometimes called ``negative f\/lows'' of integrable
hierarchies.

Nevertheless the approach of \cite{FNR,Hol1,Hol2} does not cover
all integrable equations. In particular, it was not possible to
produce by means of this approach integrable equations possessing
$U$-$V$ pairs with $U$-operator belonging to
$\widetilde{\mathfrak{g}}_{+}$ and $V$-operator belonging to
$\widetilde{\mathfrak{g}}_{-}$ in the case
$\widetilde{\mathfrak{g}}_{+}\cap \widetilde{\mathfrak{g}}_{-}\neq
0$, i.e. in  cases dropping out of the scope of
Kostant--Adler--Symes method. An example of such situation is the
Thirring integrable equations with the standard $sl(2)$-valued
$U$-$V$ pair \cite{Mikh}.

In the present paper we generalize the method of \cite{FNR,New} of
producing soliton equations and their $U$-$V$ pairs, f\/illing the
gap described above, i.e. making the method to include all known
soliton equations, and among them Thirring equation. We use the
same idea as in  \cite{FNR} and \cite{Hol1}, i.e.\ we interpret
zero-curvature condition as a compatibility condition for a set of
commuting Hamiltonian f\/lows on $\widetilde{\mathfrak{g}}^*$ but
constructed not with the help of Kostant--Adler--Symes method but
with the help of its generalization -- method of the classical
$R$-operator, where $R:\widetilde{\mathfrak{g}} \rightarrow
\widetilde{\mathfrak{g}}$,
 satisf\/ies a modif\/ied Yang--Baxter equation~\cite{ST}. In this case one
can also def\/ine \cite{ST} two Lie
subalgebras~$\widetilde{\mathfrak{g}}_{R_{\pm}}$ such that
$\widetilde{\mathfrak{g}}_{R_{\pm}}+\widetilde{\mathfrak{g}}_{R_{-}}=\widetilde{\mathfrak{g}}$,
but in this case this sum is  not  a direct sum of two vector
spaces, i.e.\
$\widetilde{\mathfrak{g}}_{R_{\pm}}\cap\widetilde{\mathfrak{g}}_{R_{-}}\neq
0$.

Hence, in order to achieve our goal it is necessary to construct
with the help of the $R$-operator an algebra of mutually commuting
functions. Contrary to the classical approach of~\cite{ST} they
should commute not with respect to the $R$-bracket $\{\; ,\; \}_R$
but with respect to the initial Lie--Poisson bracket $\{\; ,\; \}$
on $\widetilde{\mathfrak{g}}^*$. In our previous paper \cite{Skr9}
the corresponding functions were constructed using the ring of
Casimir functions $I^G(\tilde{\mathfrak{g}}^*)$. In more detail,
we proved that the functions $ I^{R_{\pm}}_k(L)\equiv I_k((R^* \pm
1 )(L))$,
 where
$I_k(L), I_l(L)\in I^G(\widetilde{\mathfrak{g}}^*)$,  constitute
an Abelian subalgebra in $C^{\infty}(\widetilde{\mathfrak{g}}^*)$
with respect to the standard Lie--Poisson brackets $\{\;,\;\}$ on
$\widetilde{\mathfrak{g}}^*$. The algebra-valued gradients of
functions the $I^{R_{\pm}}_k(L)$ belong to the subalgebras
$\widetilde{\mathfrak{g}}_{R_{\pm}}$ correspondingly.

In the case when the $R$-operator is of Kostant--Adler--Symes
type, i.e.~$R=P_+-P_-$, where $P_{\pm}$ are projection operators
onto subalgebras
$\widetilde{\mathfrak{g}}_{R_{\pm}}=\widetilde{\mathfrak{g}}_{\pm}$
and
$\widetilde{\mathfrak{g}}_{+}\cap\widetilde{\mathfrak{g}}_{-}=0$
we re-obtain the results of~\cite{New} (see also~\cite{Skr8}) as a
partial case of our construction. In the cases of more complicated
$R$-operators our scheme is new and generalizes the approach
of~\cite{New}. In particular, the important class of $U$-$V$ pairs
satisfying zero-curvature equations that can be obtained by our
method are connected with the so-called ``triangular''
$R$-operators. In more detail, if  $\widetilde{\mathfrak{g}}$
possess a ``triangular'' decomposition: $
\widetilde{\mathfrak{g}}=\widetilde{\mathfrak{g}}_+ +
\mathfrak{g}_0 + \widetilde{\mathfrak{g}}_- $, where the sum is a
direct sum of vector spaces, $\widetilde{\mathfrak{g}}_{\pm}$ and
$\mathfrak{g}_0$ are closed Lie subalgebras and
$\widetilde{\mathfrak{g}}_{\pm}$ are $\mathfrak{g}_0$-modules,
  $R_0$ is a solution of the modif\/ied
Yang--Baxter equation  on $\mathfrak{g}_0$, $P_{\pm}$ are the
projection operators onto the subalgebras
$\widetilde{\mathfrak{g}}_{\pm}$ then $ R=P_+ + R_0 -P_- $
 is a solution of the modif\/ied Yang--Baxter equation on
$\widetilde{\mathfrak{g}}$ (see \cite{Skr9} for the detailed
proof). The Lie subalgebras $\widetilde{\mathfrak{g}}_{R_{\pm}}$
have in this case the following form:
$\widetilde{\mathfrak{g}}_{R_{\pm}}=\widetilde{\mathfrak{g}}_{\pm}+
\mathfrak{g}_0$, i.e. $\widetilde{\mathfrak{g}}_{R_{+}}\cap
\widetilde{\mathfrak{g}}_{R_{-}} \neq 0$. Such $R$-operators are
connected with the Thirring-type integrable models. These
$R$-operators were f\/irst considered in \cite{Guil} and in
\cite{GM}, where the usual $sl(2)$-Thirring equation was obtained
using ``geometric'' technique.

In order to obtain the Thirring integrable equation and its
various generalizations in the framework of our algebraic approach
we consider the case when $\widetilde{\mathfrak{g}}$ coincides
with a loop algebra $\widetilde{\mathfrak{g}}= \mathfrak{g}\otimes
{\rm Pol}(u^{-1},u)$ or its ``twisted'' with the help of
f\/inite-order automorphism $\sigma$
subalgeb\-ra~$\widetilde{\mathfrak{g}}^{\sigma}$. The algebras
$\widetilde{\mathfrak{g}}^{\sigma}$ possess the natural
``triangular'' decomposition with the algebra $\mathfrak{g}_0$
being a~reductive subalgebra of $\mathfrak{g}$ stable under the
action of $\sigma$. For each algebra
$\widetilde{\mathfrak{g}}^{\sigma}$ with a natural triangular
decomposition and for each classical $R$ operator $R_0$ on
$\mathfrak{g}_0$ we def\/ine an integrable equation of the
hyperbolic type which we call the ``non-Abelian generalized
Thirring equation''. We show that in the case
$\mathfrak{g}=sl(2)$, $\sigma^2=1$ and
$\mathfrak{g}_0=\mathfrak{h}$, where $\mathfrak{h}$ is a Cartan
subalgebra of~$sl(2)$ our construction yields the usual Thirring
equation and its standard $sl(2)$-valued $U$-$V$ pair. We consider
in some detail the cases of the generalized Thirring equations
that correspond to the second order automorphism of
$\mathfrak{g}$. For the case of such automorphisms and
$\mathfrak{g}=gl(n)$ and special choice of $R_0$ we obtain the
so-called matrix generalization of complex Thirring equations
obtained by other method in~\cite{TW}. After a reduction to  real
subalgebra $u(n)$ these equations read as follows:
\begin{gather*}
i\partial_{x_-} \Psi_+=\left(\frac{1}{\kappa_-}\Psi_+
(\Psi^{\dag}_{-}\Psi_{-}) +\kappa_+\Psi_{-}\right), \qquad
i\partial_{x_+}
\Psi_-=-\left(\frac{1}{\kappa_+}(\Psi_{+}\Psi^{\dag}_{+})\Psi_- +
\kappa_-\Psi_+\right),
\end{gather*}
where $\Psi_{\pm} \in {\rm Mat}(p,q)$, $\kappa_{\pm}\in
\mathbb{R}$ are constants, i.e.\ are exact matrix analogues of the
usual massive Thirring equations.

We also consider in detail the case of the generalized Thirring
equations that correspond to the Coxeter automorphisms
of~$\mathfrak{g}$.
 We call the generalized Thirring equations
corresponding to the the Coxeter automorphisms the ``generalized
Abelian Thirring equations''. We show that the number of
independent f\/ields in the generalized Abelian Thirring equation
corresponding to Coxeter automorphism is equal to  $2
(\mathrm{dim}\, \mathfrak{g}- \mathrm{rank}\, \mathfrak{g})$ and
the order of the equations grows with the growth of the rank of
Lie algebra $\mathfrak{g}$. We consider in detail the generalized
Abelian Thirring equations corresponding to the case
$\mathfrak{g}=sl(3)$.

For the sake of completeness we also consider  non-linear
dif\/ferential equations of hyperbolic type corresponding to
$\widetilde{\mathfrak{g}}^{\sigma}$ and the Kostant--Adler--Symes
$R$-operator. We show that the obtained equations are in a sence
intermediate between the generalized Thirring and non-Abelian Toda
equations. The cases of the second order  and Coxeter
automorphisms are considered. The $sl(2)$ and $sl(3)$ examples are
worked out in detail.

The structure of the present paper is the following: in the second
section we describe commutative subalgebras  associated with the
classical $R$ operator. In the third section we obtain associated
zero-curvature equations. At last in the fourth section we
consider integrable hierarchies associated with loop algebras in
dif\/ferent gradings and dif\/ferent $R$-operators.

\section[Commutative subalgebras  and classical $R$ operator]{Commutative subalgebras  and classical $\boldsymbol{R}$ operator}

Let $\widetilde{\mathfrak{g}}$ be a Lie algebra (f\/inite or
inf\/inite-dimensional) with a Lie bracket $[\ ,\ ]$. Let
$R:\widetilde{\mathfrak{g}}\rightarrow \widetilde{\mathfrak{g}}$
be some linear operator on $ \widetilde{\mathfrak{g}}$. The
operator $R$ is called the classical $R$-operator if it
satisf\/ies modif\/ied Yang--Baxter equation:
\begin{equation}\label{mybe}
R([R(X),Y]+ [X,R(Y)])-[R(X),R(Y)]=[X,Y]\qquad  \forall\, X,Y \in
\widetilde{\mathfrak{g}}.
\end{equation}
We will  use, in addition  to the operator $R$, the following
operators: $R_{\pm}\equiv R\pm 1$.  As it is known~\cite{BD,ST}
the maps $R_{\pm}$ def\/ine Lie subalgebras
$\widetilde{\mathfrak{g}}_{R_{\pm}}\subset
\widetilde{\mathfrak{g}}$:
$\widetilde{\mathfrak{g}}_{R_{\pm}}=\mathrm{Im}\, R_{\pm}$. It is
easy to see from their def\/inition
$\widetilde{\mathfrak{g}}_{R_{+}}+\widetilde{\mathfrak{g}}_{R_{-}}=\widetilde{\mathfrak{g}}$,
but, in general, this sum is  not a direct sum of vector spaces,
i.e.,
$\widetilde{\mathfrak{g}}_{R_{+}}\cap\widetilde{\mathfrak{g}}_{R_{-}}\neq
0$.

Let $\widetilde{\mathfrak{g}}^*$  be the dual space to
$\widetilde{\mathfrak{g}}$ and
 $\langle \ , \ \rangle:
\widetilde{\mathfrak{g}}^*\times
\widetilde{\mathfrak{g}}\rightarrow \mathbb{C}$ be a pairing
between $\widetilde{\mathfrak{g}}^*$ and
$\widetilde{\mathfrak{g}}$. Let~$\{X_i\}$ be a basis in the Lie
algebra $\widetilde{\mathfrak{g}}$, $\{X^*_i\}$ be a basis in the
dual space $\widetilde{\mathfrak{g}}^*$: $\langle X^*_j, X_i
\rangle =\delta_{ij}$, $L=\sum\limits_{i} L_iX^*_i \in
\widetilde{\mathfrak{g}}^*$ be the generic element of
$\widetilde{\mathfrak{g}}^*$, $L_i$ be the coordinate functions on
$\widetilde{\mathfrak{g}}^*$.

Let us consider the standard Lie--Poisson bracket between $F_1,
F_2 \in C^{\infty}(\widetilde{\mathfrak{g}}^*)$:
\begin{equation*}\label{lpb}
\{F_1(L),F_2(L)\}= \langle L , [\nabla F_1, \nabla F_2] \rangle,
\end{equation*}
where
\begin{equation*}%\label{mg}
\nabla F_k(L)= \sum\limits_{i} \frac{\partial F_k(L)}{\partial
L_i}X_i
\end{equation*}
is a so-called algebra-valued gradient of $F_k$. The summations
hereafter are implied over all basic elements
of~$\widetilde{\mathfrak{g}}$.

 Let $R^*$ be the operator dual to $R$, acting in the space
$\widetilde{\mathfrak{g}}^*$: $ \langle R^{*}(Y),X\rangle\equiv
\langle Y,R(X)\rangle$, $\forall\, Y \in
\widetilde{\mathfrak{g}}^*$, $X \in \widetilde{\mathfrak{g}}$. Let
$I^{G}(\widetilde{\mathfrak{g}}^*)$ be the ring of invariants of
the coadjoint representation of $\widetilde{\mathfrak{g}}$.

We will consider the functions $I^{R_{\pm}}_k(L)$ on
$\widetilde{\mathfrak{g}}^*$ def\/ined by the following formulas:
\begin{equation*}
I^{R_{\pm}}_k(L)\equiv  I_k((R^* \pm 1 )(L))\equiv I_k(R^*_{\pm}
(L)).
\end{equation*}

The following theorem holds true \cite{Skr9}:
\begin{theorem}\label{mainteo}
Functions $\{I^{R_{+}}_k(L)\}$ and $\{I^{R_{-}}_l(L) \}$, where
$I_k(L), I_l(L)\in I^G(\widetilde{\mathfrak{g}}^*)$, generate an
Abelian subalgebra in $C^{\infty}(\widetilde{\mathfrak{g}}^*)$
with respect to the standard Lie--Poisson brackets $\{\; ,\; \}$
on $\widetilde{\mathfrak{g}}^*$:
\begin{gather*}
(i) \  \{I^{R_{+}}_k(L),I^{R_{+}}_l(L)\}=0, \quad (ii) \
\{I^{R_{-}}_k(L),I^{R_{-}}_l(L)\}=0, \quad (iii) \
\{I^{R_{+}}_k(L),I^{R_{-}}_l(L)\}=0.
\end{gather*}
\end{theorem}

\begin{remark} Note that the commutative subalgebras constructed
in this theorem  dif\/fer from the commutative subalgebras
constructed using standard $R$-matrix scheme \cite{ST}. Indeed,
our theorem states commutativity of functions $\{I^{R_{+}}_k(L)\}$
and $\{I^{R_{-}}_l(L) \}$ with respect to the initial Lie--Poisson
bracket $\{\ ,\ \}$ on $\widetilde{\mathfrak{g}}$. The standard
$R$-matrix scheme states commutativity of the functions
$\{I_k(L)\}$ with respect to the so-called $R$-bracket $\{\ ,\
\}_R$, where:
\begin{equation*}
\{F_1(L),F_2(L)\}_R= \langle L , [R(\nabla F_1), \nabla F_2]+
[\nabla F_1, R(\nabla F_2)] \rangle.
\end{equation*}
\end{remark}

Theorem \ref{mainteo} provides us a large Abelian subalgebra in
the space $(C^{\infty}(\widetilde{\mathfrak{g}}^*),\{\; ,\; \})$.
We will consider the following two  examples of the $R$-operators
and the corresponding Abelian subalgebras.

\begin{example}
 Let us consider the  case of the Lie algebras
$\widetilde{\mathfrak{g}}$ with the so-called
``Kostant--Adler--Symes'' (KAS) decomposition into a direct sum of
the two vector subspaces $\widetilde{\mathfrak{g}}_{\pm}$:
\[
\widetilde{\mathfrak{g}}=\widetilde{\mathfrak{g}}_+ +
\widetilde{\mathfrak{g}}_-,
\]
where subspaces $\widetilde{\mathfrak{g}}_{\pm}$ are closed Lie
subalgebras. Let $P_{\pm}$ be the projection operators onto the
subalgebras $\widetilde{\mathfrak{g}}_{\pm}$ respectively. Then it
is known \cite{ST} that in this case it is possible to def\/ine
the so-called Kostant--Adler--Symes $R$-matrix:
\[
R=P_+ - P_-.
\]
 It is easy to see that  $R_+=1+R=2P_+$,
$R_-=R-1=-2P_-$, are proportional to the projection operators onto
the  subalgebras $\widetilde{\mathfrak{g}}_{\pm}$. It follows that
$\widetilde{\mathfrak{g}}_{R_{\pm}}\equiv
\widetilde{\mathfrak{g}}_{\pm}$ and
$\widetilde{\mathfrak{g}}_{R_{+}} \cap
\widetilde{\mathfrak{g}}_{R_{-}}=0$.

The Poisson commuting functions $I^{R_{\pm}}_k(L)$  acquire the
following simple form:
\begin{equation*}%\label{ka-i}
I^{R_{\pm}}_k(L)\equiv  I_k^{\pm}(L)\equiv I_k(L_{\pm}), \qquad
\text{where}\quad L_{\pm} \equiv P^*_{\pm} L,
\end{equation*}
i.e.\ $I_k^{\pm}(L)$ are restrictions of the coadjoint invariants
 onto the dual spaces
$\widetilde{\mathfrak{g}}_{\pm}^*$.
\end{example}

\begin{example}  Let us consider the  case of  Lie algebras
$\widetilde{\mathfrak{g}}$ with the ``triangular'' decomposition:
\begin{equation*}%\label{tiang}
\widetilde{\mathfrak{g}}=\widetilde{\mathfrak{g}}_+ +
\mathfrak{g}_0 + \widetilde{\mathfrak{g}}_-,
\end{equation*}
where the sum is a direct sum of vector spaces,
$\widetilde{\mathfrak{g}}_{\pm}$ and $\mathfrak{g}_0$ are closed
subalgebras, and $\widetilde{\mathfrak{g}}_{\pm}$ are
$\mathfrak{g}_0$-modules. As it is known \cite{Guil} (see also
\cite{Skr9} for the detailed proof),
 if $R_0$ is a solution of the modif\/ied
Yang--Baxter equation (\ref{mybe}) on $\mathfrak{g}_0$ then
\begin{equation}\label{trm}
R=P_+ + R_0 -P_-
\end{equation}
is a solution of the modif\/ied Yang--Baxter (\ref{mybe}) equation
on $\mathfrak{g}$ if $P_{\pm}$ are the projection operators onto
the subalgebras $\widetilde{\mathfrak{g}}_{\pm}$ .

 In the case when $R_0=\pm Id_0$ (which are
obviously the solutions of the equation (\ref{mybe}) on
$\mathfrak{g}_0$) we obtain that $R$-matrix (\ref{trm}) passes to
the standard Kostant--Adler--Symes $R$-matrix.  Nevertheless in
the considered ``triangular cases'' there are other possibilities.
For example, if  a Lie subalgebra $\mathfrak{g}_0$ is Abelian then
(\ref{trm}) is a solution of (\ref{mybe}) for any tensor $R_0$ on
$\mathfrak{g}_0$.

 In the case of the $R$-matrix (\ref{trm}) we have:
$R_{+}= 2P_+ +(P_0 + R_0)$, $R_{-}= -(2P_- +(P_0 - R_0))$,
$\widetilde{\mathfrak{g}}_{R_{\pm}}=\widetilde{\mathfrak{g}}_{\pm}
+ {\rm Im}\, (R_0)_{\pm}$, where $(R_0)_{\pm}=(R_0 \pm P_0)$ are
the $R_{\pm}$-operators on $\mathfrak{g}_0$ and
$\widetilde{\mathfrak{g}}_{R_{+}}\cap
\widetilde{\mathfrak{g}}_{R_{-}}={\rm Im}\, (R_0)_{+}\cap {\rm
Im}\, (R_0)_{-}$. The Poisson-commutative functions
 $I^{R_{\pm}}_k(L)$ acquire  the following form:
\begin{equation*}
I^{R_{\pm}}_k(L)\equiv I^{R_{0,\pm}}_k(L)= I_k\left(L_{\pm}
+\frac{(1 \pm R^*_0)}{2}(L_0)\right),
\end{equation*}
where $L_{\pm} \equiv P^*_{\pm} L$, $L_{0} \equiv P^*_{0} L$. We
will use such functions constructing dif\/ferent generalizations
of the Thirring model.
\end{example}

The  Theorem \ref{mainteo} gives us a  set of mutually commutative
functions on $\widetilde{\mathfrak{g}}^*$ with respect to the
brackets $\{\; ,\; \}$ that can be used as an algebra of the
integrals of some Hamiltonian system
on~$\widetilde{\mathfrak{g}}^*$. For the Hamiltonian function one
may chose one of the functions $I_k^{R_{\pm}}(L)$ or their linear
combination. Let us consider the corresponding Hamiltonian
equation:
\begin{equation}\label{hameq}
\frac{d L_i}{d t^{\pm}_k}=\{L_i,I_k^{R_{\pm}}(L)\}.
\end{equation}
The following proposition is true \cite{Skr9}:

\begin{proposition}
The Hamiltonian equations of motion \eqref{hameq} can be written
in the  Euler--Arnold (generalized Lax) form:
\begin{equation}\label{glax}
\frac{d L}{d t^{\pm}_k}={\rm ad}^*_{V_k^{\pm}} L,
\end{equation}
where $V_k^{\pm} \equiv \nabla I_k^{R_{\pm}}(L)$.
\end{proposition}

\begin{remark} In the case when  $\widetilde{\mathfrak{g}}^*$ can
be identif\/ied with $\widetilde{\mathfrak{g}}$ and a coadjoint
representation can be identif\/ied with an adjoint one,  equation
(\ref{glax}) may be written in the usual Lax (commutator) form.
\end{remark}

In the present paper we will not consider f\/inite-dimensional
Hamiltonian systems that could be obtained in the framework of our
construction but will use equation (\ref{glax}) in order to
generate hierarchies of soliton equations in $1+1$ dimensions.

\section[Integrable hierarchies and classical $R$-operators]{Integrable hierarchies and classical $\boldsymbol{R}$-operators}

%\subsection{General case }
 Let us remind one of the main Lie algebraic
approaches to the theory of soliton equations \cite{New}. It is
based on the  zero-curvature condition and its interpretation as a
consistency condition of two commuting Lax f\/lows. The following
Proposition is true:
\begin{proposition}\label{mzc}
Let $H_1, H_2 \in C^{\infty}(\widetilde{\mathfrak{g}}^*)$ be two
Poisson commuting functions: $\{H_1, H_2\}=0$, where $\{\; ,\;
\}$ is a standard Lie--Poisson brackets on
$\widetilde{\mathfrak{g}}^*$. Then their
$\widetilde{\mathfrak{g}}$-valued gradients satisfy the
``modified'' zero-curvature equation:
\begin{equation}\label{modzce}
\frac{\partial \nabla H_1}{\partial t_2} - \frac{\partial \nabla
H_2}{\partial t_1} + [\nabla H_1, \nabla H_2]= k \nabla I,
\end{equation}
where $I$ is a Casimir function and $t_i$ are parameters along the
trajectories of Hamiltonian equations corresponding to the
Hamiltonians $H_i$ and $k$ is an arbitrary constant.
\end{proposition}

For the theory of soliton equations one  needs the usual
zero-curvature condition (case \mbox{$k=0$} in the above
proposition) and the inf\/inite set of the commuting Hamiltonians
generating the corresponding $U$-$V$ pairs. This can be achieved
requiring that $\widetilde{\mathfrak{g}}$ is
inf\/inite-dimensional of a~special type. In more detail, the
following theorem is true:

\begin{theorem}
Let  $\widetilde{\mathfrak{g}}$ be an infinite-dimensional Lie
algebra of $\mathfrak{g}$-valued function of the one complex
variable $u$ and a Lie algebra $\mathfrak{g}$ be semisimple. Let
$R$ be a classical $R$ operator on $\widetilde{\mathfrak{g}}$,
$L(u)$ be the generic element of the dual space
$\widetilde{\mathfrak{g}}^*$, $I_k(L(u))$ be  Casimir functions on
$\widetilde{\mathfrak{g}}^*$ such that functions
$I_k^{R_{\pm}}(L(u))$ are finite polynomials on
$\widetilde{\mathfrak{g}}^*$.  Then the
$\widetilde{\mathfrak{g}}$-valued functions  $\nabla
I_k^{R_{\pm}}(L(u))$ satisfy the zero-curvature equations:
\begin{gather}\label{zce1}
\frac{\partial \nabla I_k^{R_{\pm}}(L(u))}{\partial t^{\pm}_l} -
\frac{\partial \nabla I_l^{R_{\pm}}(L(u))}{\partial t^{\pm}_k} +
[\nabla I_k^{R_{\pm}}(L(u)), \nabla I_l^{R_{\pm}}(L(u))]=0,
\\ \label{zce2}
\frac{\partial \nabla I_k^{R_{\pm}}(L(u))}{\partial t^{\mp}_l} -
\frac{\partial \nabla I_l^{R_{\mp}}(L(u))}{\partial t^{\pm}_k} +
[\nabla I_k^{R_{\pm}}(L(u)), \nabla I_l^{R_{\mp}}(L(u))]=0.
\end{gather}
\end{theorem}

\begin{proof} As it follows from  Proposition~\ref{mzc} and Theorem \ref{mainteo} the algebra-valued
gradients\linebreak  $\nabla I_k^{R_{\pm}}(L(u))$, $\nabla
I_l^{R_{\pm}}(L(u))$ and $\nabla I_k^{R_{\pm}}(L(u))$, $\nabla
I_l^{R_{\mp}}(L(u))$ satisfy the ``modif\/ied'' zero-curvature
equations (\ref{modzce}). On the other hand, as it is not
dif\/f\/icult to show, for the case of the Lie algebras
$\widetilde{\mathfrak{g}}$ described in the theorem the
algebra-valued gradients of the Casimir functions are proportional
to powers of the generic element of the dual space $L(u)$, i.e.\
to the formal power series. On the other hand,  due to the
condition that all $I_k^{R_{\pm}}(L(u))$ are f\/inite polynomials,
their algebra-valued gradients are {\it finite} linear
combinations of the basic elements of the Lie algebra
$\widetilde{\mathfrak{g}}$. Hence the corresponding modif\/ied
zero-curvature equations are satisf\/ied if and only if the
corresponding coef\/f\/icients~$k$ in these equations are equal to
zero, i.e.\ when they are reduced to the usual zero-curvature
conditions. This proves the theorem.
\end{proof}

\begin{remark} Note that equations (\ref{zce1}), (\ref{zce2})
def\/ine three types of integrable hierarchies: two ``small''
hierarchies associated with Lie subalgebras
$\widetilde{\mathfrak{g}}_{R_{\pm}}$ def\/ined by equations
(\ref{zce1}) and one ``large'' hierar\-chy associated with the whole
Lie algebra $\widetilde{\mathfrak{g}}$ that include both types of
equations~(\ref{zce1}) and~(\ref{zce2}). Equations (\ref{zce2})
have an interpretation of the ``negative f\/lows'' of the
integrable hierarchy associated with
$\widetilde{\mathfrak{g}}_{R_{\pm}}$. In the case
$\widetilde{\mathfrak{g}}_{R_{+}}\simeq
\widetilde{\mathfrak{g}}_{R_{-}}$ the corresponding ``small''
hierarchies are  equivalent and the large hierarchy associated
with $\widetilde{\mathfrak{g}}$ may be called a ``double'' of the
hierarchy associated with~$\widetilde{\mathfrak{g}}_{R_{\pm}}$.
\end{remark}

In the next subsection we will consider in detail a simple example
of the above theorem when $\widetilde{\mathfrak{g}}$ is a loop
algebra and $R$-operator is not of Kostant--Adler--Symes type but
a triangular one. We will be interested in the ``large'' hierarchy
associated with  $\widetilde{\mathfrak{g}}$ and, in more detail,
in the simplest equation of this hierarchy  which  will coincide
with the generalization of the Thirring equation.

\section{Loop algebras and generalized Thirring equations}
\subsection[Loop algebras and classical $R$-operators]{Loop algebras and classical $\boldsymbol{R}$-operators}\label{loop}

In this subsection we  remind several important facts from the
theory of loop algebras~\cite{Kac}.

 Let $\mathfrak{g}$ be semisimple
(reductive) Lie algebra.
  Let
$\mathfrak{g}=\sum\limits_{j=0}^{p-1} \mathfrak{g}_{\overline{j}}$
be $Z_p=\mathbb{Z}/p \mathbb{Z}$ grading of $\mathfrak{g}$, i.e.:
$[\mathfrak{g}_{\overline{i}},\mathfrak{g}_{\overline{j}}]\subset
\mathfrak{g}_{\overline{i+j}}$ where $\overline{j}$ denotes the
class of equivalence of the elements $j\in \mathbb{Z}$
$\mathrm{mod} p\ \mathbb{Z}$. It is known that the
 $\mathbb{Z}_p$-grading of $\mathfrak{g}$ may be
 def\/ined with the help of some automorphism
 $\sigma$ of the  order $p$, such that
 $\sigma(\mathfrak{g}_{\overline{i}})= e^{2\pi i k/p}
 \mathfrak{g}_{\overline{i}}$ and $\mathfrak{g}_{\overline{0}}$ is the algebra of
$\sigma$-invariants:
$\sigma(\mathfrak{g}_{\overline{0}})=\mathfrak{g}_{\overline{0}}$.

Let $\widetilde{\mathfrak{g}}=\mathfrak{g}\otimes {\rm Pol}(u,
u^{-1})$ be a  loop algebra.
 Let us
consider the following subspace in $\widetilde{\mathfrak{g}}$:
\begin{equation*}
\widetilde{\mathfrak{g}}^{\sigma}=\bigoplus\limits_{j\in
\mathbb{Z}}\mathfrak{g}_{\overline{j}} \otimes u^j.
\end{equation*}
It is known \cite{Kac} that this subspace is a closed Lie
subalgebra and if we extend the automorphism~$\sigma$ to the map
$\tilde{\sigma}$ of the whole algebra $\widetilde{\mathfrak{g}}$,
def\/ining its action on the space $\mathfrak{g}\otimes {\rm
Pol}(u, u^{-1})$ in the standard way \cite{Kac}:
$\tilde{\sigma}(X\otimes u^k)=\sigma(X)\otimes e^{-2\pi i k/p}
u^k$, then the subalgebra
 $\widetilde{\mathfrak{g}}^{\sigma}$ can be def\/ined  as  the subalgebra of
  $\tilde{\sigma}$-invariants in $\widetilde{\mathfrak{g}}$:
\[
\widetilde{\mathfrak{g}}^{\sigma}=\{X\otimes
p(u)\in\widetilde{\mathfrak{g}}|\ \tilde{\sigma}( X\otimes p(u))=
X\otimes p(u)\}.
\]

We will call the algebra $\widetilde{\mathfrak{g}}^{\sigma}$ the
loop subalgebra ``twisted''  with the help of $\sigma$. The basis
in~$\widetilde{\mathfrak{g}}^{\sigma}$ consists of algebra-valued
functions $\{X^j_{\alpha}\equiv X^{\overline{j}}_{\alpha} u^j \}$,
where $X^{\overline{j}}_{\alpha}\in \mathfrak{g}_{\overline{j}}$.
Let us def\/ine the pairing between
$\widetilde{\mathfrak{g}}^{\sigma}$ and
$(\widetilde{\mathfrak{g}}^{\sigma})^*$ in the standard way:
\begin{equation*}%\label{pai'}
\langle X,Y \rangle=\mathop{\rm
res}\limits_{u=0}u^{-1}(X(u),Y(u)),
\end{equation*}
where $X\in \widetilde{\mathfrak{g}}^{\sigma}$, $Y \in
(\widetilde{\mathfrak{g}}^{\sigma})^*$ and $(\ ,\ )$ is a
bilinear, invariant, nondegenerate form on $\mathfrak{g}$. It is
easy to see that with respect to such a pairing the
 dual space
$(\widetilde{\mathfrak{g}}^{\sigma})^*$ may be identif\/ied with
the Lie algebra $\widetilde{\mathfrak{g}}^{\sigma}$ itself. The
dual basis in $(\widetilde{\mathfrak{g}}^{\sigma})^*$ has the
form: $\{Y^j_{\alpha}\equiv X^{\overline{-j},\alpha} u^{-j} \}$,
where $X^{\overline{-j},\alpha}$ is a~dual basis in the space
$\mathfrak{g}_{\overline{-j}}$.

 The Lie algebra
$\widetilde{\mathfrak{g}}^{\sigma}$ possesses  KAS decomposition
$\widetilde{\mathfrak{g}}^{\sigma}=\widetilde{\mathfrak{g}}^{\sigma+}+
\widetilde{\mathfrak{g}}^{\sigma-}$ \cite{RST}, where
\begin{equation*}
\widetilde{\mathfrak{g}}^{\sigma+}=\bigoplus\limits_{j \geq 0
}\mathfrak{g}_{\overline{j}} \otimes u^j, \qquad
\widetilde{\mathfrak{g}}^{\sigma-}=\bigoplus\limits_{j < 0
}\mathfrak{g}_{\overline{j}} \otimes u^j.
\end{equation*}
It def\/ines in a natural way the Kostant--Adler--Symes
$R$-operator:
\[
R=P^+-P^-,
\]
 where  $P^{\pm}$ are projection operators onto Lie
algebra $\widetilde{\mathfrak{g}}^{\sigma \pm}$.

The twisted loop  algebra $\widetilde{\mathfrak{g}}^{\sigma}$ also
possesses ``triangular'' decomposition:
\[
\widetilde{\mathfrak{g}}^{\sigma}=\widetilde{\mathfrak{g}}^{\sigma}_+
+ {\mathfrak{g}}_0 + \widetilde{\mathfrak{g}}^{\sigma}_-,
\]
where
$\widetilde{\mathfrak{g}}^{\sigma-}\equiv\widetilde{\mathfrak{g}}^{\sigma}_-$,
$\widetilde{\mathfrak{g}}^{\sigma+} \equiv
\widetilde{\mathfrak{g}}^{\sigma}_+ + {\mathfrak{g}}_0$.

It def\/ines in a natural way the triangular $R$-operator
\[
R=P_+ +R_0-P_-,
\]
where  $P_{\pm}$ are the projection operators onto Lie algebra
$\widetilde{\mathfrak{g}}_{\sigma \pm}$ and $R_0$ is an
$R$-operator on ${\mathfrak{g}}_0$.

The Lie algebra ${\mathfrak{g}}_0$ in this decomposition is a
reductive Lie subalgebra of the Lie algebra~${\mathfrak{g}}$. Due
to the fact that a lot of solutions of the modif\/ied classical
Yang--Baxter equations on reductive Lie algebras are known,  one
can construct explicitly $R$-operators $R_0$ on
${\mathfrak{g}}_0$, the Lie subalgebras
$\widetilde{\mathfrak{g}}_{R_{\pm}}=\widetilde{\mathfrak{g}}_{\pm}
+ {\rm Im} (R_0)_{\pm}$ and the Poisson-commutative functions
constructed in the Theorem~\ref{mainteo}.

\subsection{Generalized Thirring models}
As we have seen in the previous subsection, each of the gradings
of the loop algebras, corresponding to the dif\/ferent
automorphisms $\sigma$ yields its own triangular decomposition.
Let us consider the simplest equations of integrable hierarchies,
corresponding to dif\/ferent triangular decompositions in more
detail.  The generic elements of the dual space to the Lie
algebras $\widetilde{\mathfrak{g}}_{R_{\pm}}$ are written as
follows:
\begin{gather*}
L_{\pm}(u)=r^*_{\pm}\left(\sum\limits_{\alpha=1}^{\mathrm{dim}
\mathfrak{g}_{\bar{0}}} L^0_{\alpha}
X^{\overline{0},\alpha}\right)+ \sum\limits_{j=\pm 1}^{\pm \infty
} \sum\limits_{\alpha=1}^{\mathrm{dim} \mathfrak{g}_{\bar{j}}}
L^{(j)}_{\alpha} X^{\overline{-j},\alpha} u^{-j}= r^*_{\pm}(
L^{(0)}) + \sum\limits_{j=\pm 1}^{\pm \infty } L^{(j)} u^{-j},
\end{gather*}
where $L^{(-j)}\in  \mathfrak{g}_{\bar{j}}$ and in order to
simplify the notations we put $r_{\pm}\equiv \frac{(1\pm
R_0)}{2}$.

Let $(\ ,\ )$ be an invariant nondegenerated form on the
underlying semisimple (reductive) Lie algebra $\mathfrak{g}$. Then
$I_2(L)=\frac{1}{2} (L ,L )$ is a second order Casimir function on
$\mathfrak{g}^*$. The corresponding generating function of the
second order integrals has the form:
\begin{gather*}
I^{\pm}_2(L(u))= I_2(L_{\pm}(u))= \frac{1}{2} \left( r_{\pm}^*(
L^{(0)}) + \sum\limits_{j=\pm 1}^{\pm \infty } L^{(j)} u^{-j},
r_{\pm}^*( L^{(0)}) + \sum\limits_{j=\pm 1}^{\pm \infty } L^{(j)}
u^{-j} \right)
\\ \phantom{I^{\pm}_2(L(u))= I_2(L_{\pm}(u))}{}
 =\sum\limits_{k=0}^{\pm \infty } I_2^{pk}u^{-pk},
\end{gather*}
where $p$ is an order of  $\sigma$. The simplest of these
integrals are:
\begin{gather*}
 I_2^{\pm 0}= \frac{1}{2} \bigl( r_{\pm}^*( L^{(0)}),
r_{\pm}^*(L^{(0)}) \bigr),\qquad
  I_2^{\pm p}= \bigl( r_{\pm}^* (L^{(0)}), L^{(\pm p)} \bigr) +
  \frac{1}{2} \sum\limits_{j=1}^{p-1}
\bigl(  L^{(\pm j)}, L^{(\pm (p-j))} \bigr).
\end{gather*}
The algebra-valued gradients of functions $I_2^{\pm p}$ read as
follows:
\begin{equation}\label{uvpt}
\nabla I_2^{\pm p}(u)= u^{\pm p} r^*_{\pm}( L^{(0)}) +
  \sum\limits_{j=1}^{p-1} u^{\pm (p-j)}
  L^{(\pm j)} + r_{\pm} ( L^{(\pm p)}).
\end{equation}
 They coincide with the  $U$-$V$ pair of the generalized Thirring model. The
corresponding zero-curvature condition:
\begin{equation}\label{zce2t}
\frac{\partial \nabla I_2^{+ p}(L(u))}{ \partial x_-} -
\frac{\partial \nabla I_2^{- p}(L(u))}{\partial x_+} + [\nabla
I_2^{+ p}(L(u)), \nabla I_2^{- p}(L(u))]=0
\end{equation}
yields the following system of dif\/ferential equations in partial
derivatives:
\begin{subequations}\label{gte}
\begin{gather} \label{gte1}
-\partial_{x_-} r^*_+(L^{(0)})=[r^*_+(L^{(0)}), r_-(L^{(-p)})],
\\
\label{gte2}
\partial_{x_+} r_-(L^{(0)})=[r_+(L^{(p)}), r^*_-(L^{(0)})],
\\
\label{gte3} -\partial_{x_-} L^{(p-j)}=[L^{(p-j)}, r_-(L^{(-p)})]
+ \sum\limits_{i=j+1}^{p-1} [L^{(p-i)}, L^{(i-j-p)}]
+[r^*_+(L^{(0)}), L^{(-j)}],
\\
\partial_{x_+} L^{(j-p)}=[ r_+(L^{(p)}), L^{(j-p)})] +
\sum\limits_{i=j+1}^{p-1} [L^{(p+j-i)}, L^{(j-p)}]\nonumber\\
\phantom{\partial_{x_+} L^{(j-p)}=}{} +[L^{(j)},r^*_-(L^{(0)})],
j\in 1,p-1,\label{gte4}
\\
\partial_{x_+} r_-(L^{(-p)}) - \partial_{x_-}r_+( L^{(p)})=[ r^*_+(L^{(0)}),
 r^*_-(L^{(0)})] +
\sum\limits_{i=1}^{p-1} [L^{(p-i)}, L^{(i-p)}]\nonumber\\
\phantom{\partial_{x_+} r_-(L^{(-p)}) - \partial_{x_-}r_+(
L^{(p)})=}{}+[r_+(L^{(p)}),r_-(L^{(-p)})].\label{gte5}
\end{gather}
\end{subequations}
These equations are {\it integrable generalization of Thirring
equations} corresponding to Lie al\-geb\-ra~$\mathfrak{g}$, its
$Z_p$-grading def\/ined with the help of the automorphism $\sigma$
of the order $p$ and the classical $R$-operator $R_0$ on
$\mathfrak{g}_{\bar 0}$. We will call this system of equations the
{\it non-Abelian generalized Thirring system}. In order to
recognize in this complicated system of hyperbolic equations
generalization of Thirring equations  we will consider several
examples.

\subsubsection{Case of the second-order automorphism}
Let us consider the case when automorphism $\sigma$ is involutive,
i.e.\ $p=2$. Then $\mathfrak{g}_{\overline{1}}=
\mathfrak{g}_{\overline{-1}}$ and  Lax matrices  $L_{\pm}(u)$ are
the following:
\begin{equation*} L_{\pm }(u)= r_{\pm}^*( L^{(0)}) + u^{\mp 1}L^{(\pm
1)}+ u^{\mp 2}L^{(\pm 2)} + \cdots ,
\end{equation*}
and the $U$-$V$ pair (\ref{uvpt}) acquire  more simple form:
\begin{equation}\label{uvpts}
\nabla I_2^{\pm 2}(u)= u^{\pm 2} r^*_{\pm}( L^{(0)}) +
  u^{\pm 1} L^{(\pm 1)} + r_{\pm} ( L^{(\pm 2)}).
\end{equation}
The corresponding zero-curvature condition yields the following
system of dif\/ferential equations in  partial derivatives:
\begin{subequations}\label{gtes}
\begin{gather} \label{gtes1}
\partial_{x_-} r^*_+(L^{(0)})=[ r_-(L^{(-2)}), r^*_+(L^{(0)})],
\\
\label{gtes2}
\partial_{x_+} r^*_-(L^{(0)})=[r_+(L^{(2)}), r^*_-(L^{(0)})],
\\
\label{gtes3}
\partial_{x_-} L^{(1)}=[r_-(L^{(-2)}), L^{(1)}] + [L^{(-1)}, r^*_+(L^{(0)})],
\\
\label{gtes4}
\partial_{x_+} L^{(-1)}=[ r_+(L^{(2)}), L^{(-1)}]
+[L^{(1)},r^*_-(L^{(0)})],
\\
\partial_{x_+} r_-(L^{(-2)})- \partial_{x_-} r_+(L^{(2)})=[ r^*_+(L^{(0)}),
 r^*_-(L^{(0)})] +
 [L^{(1)}, L^{(-1)}]\nonumber\\
\qquad{}{}+[r_+(L^{(2)}),r_-(L^{(-2)})].\label{gtes5}
\end{gather}
\end{subequations}

The system of equations (\ref{gtes}) is still suf\/f\/iciently
complicated. In order to recognize in this system the usual
Thirring equations we have to consider the
case~$\mathfrak{g}=sl(2)$.

\begin{example} Let $\mathfrak{g}=sl(2)$ and $\sigma$ be the
Cartan involution, i.e. $sl(2)_{\bar{0}}=\mathfrak{h}={\rm
diag}\,(\alpha,-\alpha)$ and $sl(2)_{\bar{1}}$ consists of the
matrices with  zeros on the diagonal. The Lax matrices
$L_{\pm}(u)$ have the following form:
\begin{gather*} L_{\pm }(u)= r_{\pm}^*  \left(%
\begin{array}{cc}
  \alpha^{(0)} & 0 \\
  0 & -\alpha^{(0)} \\
\end{array}%
\right)  + u^{\mp 1}\left(%
\begin{array}{cc}
   0 & \beta^{(\pm 1)} \\
  \gamma^{(\pm 1)} & 0  \\
\end{array}%
\right)+ u^{\mp 2}\left(%
\begin{array}{cc}
  \alpha^{(\pm 2)} & 0 \\
  0 & -\alpha^{( \pm 2)} \\
\end{array}%
\right) + \cdots .
\end{gather*}
Due to the fact that $sl(2)_{\bar{0}}$ is Abelian the
one-dimensional
 linear maps $r_{\pm}$ are written as follows: $r_{\pm}^*( L^{(0)})=k_{\pm}  L^{(0)}$ where $k_{\pm}$ are some constants,
 such that $k_{+}+ k_{-}=1$.

The simplest  Hamiltonians obtained in the framework of our scheme
are:
\begin{gather*}
 I_2^{\pm 0}=  k_{\pm} (\alpha^{(0)})^2,
\qquad
 I_2^{\pm 2}= 2k_{\pm}\alpha^{(0)}\alpha^{(\pm 2)} + \beta^{(\pm 1)}\gamma^{(\pm 1)}.
\end{gather*}
As it follows from the Theorem \ref{mainteo} and, as it is also
easy to verify by the direct calculations, these Hamiltonians
commute with respect to the standard Lie--Poisson brackets on
$\widetilde{sl(2)}^{\sigma}$:
\begin{gather*}
\{\alpha^{(2k)},\beta^{(2l+1)}\}=\beta^{(2(k+l)+1)},\qquad
\{\alpha^{(2k)},\gamma^{(2l+1)}\}=-\gamma^{(2(k+l)+1)},\\
\{\beta^{(2k+1)},\gamma^{(2l+1)}\}=2\alpha^{(2(k+l)+2)},\\
\{\alpha^{(2k)},\alpha^{(2l)}\}=
\{\beta^{(2k+1)},\beta^{(2l+1)}\}=
\{\gamma^{(2k+1)},\gamma^{(2l+1)}\}=0.
\end{gather*}
As a consequence, they produce the correct $U$-$V$ pair for
zero-curvature equations:
\begin{gather*}
 \nabla I_2^{\pm 2}= k_{\pm}\left(%
\begin{array}{cc}
  \alpha^{(\pm 2)} & 0 \\
  0 & -\alpha^{(\pm 2)} \\
\end{array}%
\right) + u^{\pm 1}\left(%
\begin{array}{cc}
   0 & \beta^{(\pm 1)} \\
  \gamma^{(\pm 1)} & 0  \\
\end{array}%
\right)+ k_{\pm}u^{\pm2}\left(%
\begin{array}{cc}
  \alpha^{(0)} & 0 \\
  0 & -\alpha^{(0)} \\
\end{array}%
\right),
\end{gather*}
where $U\equiv \nabla I_2^{+ 2}$, $V\equiv \nabla I_2^{- 2}$.

Due to the fact that $I_2^{\pm 0}$ are constants of motion we have
that $\alpha^{(0)}$ is also a constant of motion  and equations
(\ref{gtes1}), (\ref{gtes2}) are satisf\/ied automatically. The
integrals $I_2^{\pm 2}$ are constants of motion too, that is why
$\alpha^{(\pm 2)}$ are expressed via $\beta^{(\pm 1)}$,
$\gamma^{(\pm 1)}$. Hence, equation (\ref{gtes5}) is not
independent and follows from
 equations (\ref{gtes3}), (\ref{gtes4}). These equations are
independent. The last summand in the equations (\ref{gtes3}),
(\ref{gtes4}) gives the linear ``massive term'' in the Thirring
equation. The f\/irst summand gives the cubic non-linearity.

Let us write the equations (\ref{gtes3}), (\ref{gtes4}) in the
case at hand in more detail taking into account the explicit form
of the matrices $L^{(k)}$ and linear operators $r_{\pm}$:
\begin{gather}
\partial_{x_-}\gamma^{(1)}=2k_+\alpha^{(0)}\gamma^{(-1)} -2k_-\alpha^{(-2)} \gamma^{(1)},
\nonumber\\
\partial_{x_-}\beta^{(1)}=-2k_+\alpha^{(0)}\beta^{(-1)} + 2k_-\alpha^{(-2)} \beta^{(1)},
\nonumber\\
\partial_{x_+}\gamma^{(-1)}=2k_-\alpha^{(0)}\gamma^{(1)} -2k_+\alpha^{(2)} \gamma^{(-1)},
\nonumber\\
\partial_{x_+}\beta^{(-1)}=-2k_-\alpha^{(0)}\beta^{(1)} + 2k_+\alpha^{(2)} \beta^{(-1)},\label{comthir}
\end{gather}
where we have put  $I_2^{\pm 2}=0$ and, hence, $\alpha^{(\pm
2)}=-\frac{1}{2k_{\pm}\alpha^{(0)}} (\beta^{(\pm 1)}\gamma^{(\pm
1)})$.

The system of equations (\ref{comthir}) admits a reduction
$\alpha^{(0)}=ic$, $\gamma^{(\pm 1)}= -\bar{\beta}^{(\pm 1)}\equiv
-\bar{\psi}_{\pm 1}$, which corresponds to the restriction onto
the real subalgebra $su(2)$ of the complex  Lie algebra $sl(2)$.

After such a reduction the system of equations (\ref{comthir}) is
simplif\/ied and acquires the  form:
\begin{gather*}
ic\partial_{x_-}\psi_{1}=2k_+c^2\psi_{-1} + |\psi_{-1}|^2
\psi_{1},\qquad
ic\partial_{x_+}\psi_{-1}=2k_-c^2\psi_{1} + |\psi_{1}|^2 \psi_{-1}.%\label{thir}
\end{gather*}
In the case $k_+=k_-=\frac{1}{2}$ these  equations are usual
Thirring equations with a mass $m=c^2$.
\end{example}

\subsubsection{Matrix generalization of Thirring equation}

Let us return to the generalized Thirring model in the case of the
higher rank Lie algebra $\mathfrak{g}$, its automorphism of the
second order and corresponding $Z_2$-grading of $\mathfrak{g}$:
$\mathfrak{g}= \mathfrak{g}_{\bar{0}} + \mathfrak{g}_{\bar{1}}$.
We are interested in the cases that will be maximally close to the
case of the ordinary Thirring equation. In particular, we wish to
have $r^*_{\pm}(L^{(0)})$ that enter into our $U$-$V$ pair
(\ref{uvpts}) near the second order of spectral parameter to be
constant along all time f\/lows. As it follows from the equations
(\ref{gtes1}), (\ref{gtes2}),  this is not true for the general
$r$-matrix $R_0$ on $\mathfrak{g}_{\bar{0}}$. Fortunately there
are very special cases when it is indeed so. The following
proposition holds true:

\begin{proposition}
Let $\mathfrak{g}_{\bar{0}}$ admit the decomposition into direct
sum of two reductive subalgebras: $\mathfrak{g}_{\bar{0}}=
\mathfrak{g}_{\bar{0}}^+ \oplus \mathfrak{g}_{\bar{0}}^-$. Let $
L^{(0)}= L^{(0)}_++ L^{(0)}_-$ be the corresponding decomposition
of the element of the dual space. Let
$\zeta(\mathfrak{g}_{\bar{0}}^{\pm})$ be a center of the
subalgebra $\mathfrak{g}_{\bar{0}}^{\pm}$, $K(L^{(0)}_{\pm})$ be a
part of $L^{(0)}_{\pm}$ dual to the  center of the subalgebra
$\mathfrak{g}_{\bar{0}}^{\pm}$. Then $R_0=P_0^+- P^-_0$ is the $R$
operator on $\mathfrak{g}_{\bar{0}}$, $r_{+}(L^{(0)})=L^{(0)}_+$,
$ r_{-}(L^{(0)})= L^{(0)}_-$  are constant along all time flows
and the reduction $ L^{(0)}= K_++ K_-$, where $K_{\pm}$ is a
constant element of $\zeta(\mathfrak{g}_{\bar{0}}^{\pm})^*$, is
consistent with all equations of the hierarchy \eqref{zce1},
\eqref{zce2} corresponding to the loop algebra
$\widetilde{\mathfrak{g}}^{\sigma}$ and triangular $R$-operator on
$\widetilde{\mathfrak{g}}^{\sigma}$ with the described above KAS
$R$-operator $R_0$ on $\mathfrak{g}_{\bar{0}}$.
\end{proposition}

\begin{proof} Let us at f\/irst note that in the case under
consideration $r_{\pm}^*=r_{\pm}$. Let us show that
$r^*_{\pm}(L^{(0)})=r_{\pm}(L^{(0)})$  are constant along all time
f\/lows generated by the second order Hamilto\-nians~$I_2^{\pm
2k}$. It is easy to show that the corresponding algebra-valued
gradients have in this case the following form:
\[
\nabla I_2^{\pm 2k}= P_{\pm} ( u^{2k} L_{\pm}(u)) +
r_{\pm}(L^{(\pm 2k)})= u^{2k} L_{\pm}(u)-P_{\mp} ( u^{2k}
L_{\pm}(u))- r_{\mp}(L^{(\pm 2k)}) ,
\]
 where $P_{\pm}$ are the projection operators onto the
subalgebras $\widetilde{\mathfrak{g}}^{\sigma}_{\pm}$ in the
triangular decomposition of  the Lie algebra
$\widetilde{\mathfrak{g}}^{\sigma}$ and we took into account that
$r_+ + r_-=1$. Let us substitute this expression into the Lax
equation:
\begin{equation}\label{icl}
\frac{dL(u)}{d t^{\pm}_{2k}}=[ \nabla I_2^{\pm 2k}, L(u)],
\end{equation}
and take into account that the case under consideration
corresponds to the Kostant--Adler--Symes decomposition and, hence,
we have correctly def\/ined the decomposition $L(u)=L_+(u) +
L_-(u)$, and each of two inf\/inite-component Lax equation
(\ref{icl}) is correctly restricted to  each  of the subspaces
$L_{\pm}(u)$:
\begin{gather*}
\frac{dL_{\pm}(u)}{d t^{\pm}_{2k}}=[ -P_{\mp} ( u^{2k}
L_{\pm}(u))- r_{\mp}(L^{(\pm 2k)}),
 L_{\pm}(u)],\\
\frac{dL_{\pm}(u)}{d t^{\mp}_{2k}}=[ P_{\mp} ( u^{2k} L_{\mp}(u))
+ r_{\mp}(L^{(\mp 2k)}), L_{\pm}(u)].
\end{gather*}
Making projection onto the subalgebra $\mathfrak{g}_{\bar{0}}$ in
these equations we  obtain the  equations:
\begin{gather*}
\frac{d r_{\pm}( L^{(0)})}{d t^{\pm}_{2k}}=- [ r_{\mp}(L^{(\pm
2k)}), r_{\pm} (L^{(0)})],\qquad \frac{d r_{\pm}( L^{(0)})}{d
t^{\mp}_{2k}}= [ r_{\mp}(L^{(\mp 2k)}),
r_{\pm} (L^{(0)})].%\label{icl2}
\end{gather*}
Due to the fact that our $r$-matrix is of Kostant--Adler--Symes
type, we have $[r_{\mp}(X), r_{\pm} (Y)]=0$, $\forall\, X,Y \in
\mathfrak{g}_{\bar{0}}$ and, hence, $r_{+} (L^{(0)})\equiv
L^{(0)}_{+}$, $r_{-} (L^{(0)})\equiv L^{(0)}_{-}$ are constant
along all time f\/lows generated by $I_2^{\pm 2k}$. In analogous
way it is shown that $L^{(0)}= L^{(0)}_{+} + L^{(0)}_{-} $ are
constant with respect to the time f\/lows generated by the higher
order integrals $I_m^{\pm 2k}$. Hence, in this case components of
$L^{(0)}$ belong to the algebra of the integrals of motion of our
inf\/inite-component Hamiltonian or Lax system. This algebra of
integrals is, generally speaking, non-commutative i.e.\
$\{l^{(0)}_{\alpha}, l^{(0)}_{\beta} \}=c_{\alpha\beta}^{\gamma}
l^{(0)}_{\gamma}$, where $c_{\alpha\beta}^{\gamma}$ are structure
constants of the Lie algebra $\mathfrak{g}_{\bar{0}}$. That is why
in order to have a correct and consistent reduction with respect
to these integrals   we have to restrict the dynamics to the
surface of zero level of the integrals belonging to
$([\mathfrak{g}_{\bar{0}}, \mathfrak{g}_{\bar{0}}])^*$. Other part
of $(\mathfrak{g}_{\bar{0}})^*$, namely the one belonging to
$(\mathfrak{g}_{\bar{0}}/[\mathfrak{g}_{\bar{0}},
\mathfrak{g}_{\bar{0}}] )^*=
(\zeta(\mathfrak{g}_{\bar{0}}))^*=(\zeta(\mathfrak{g}_{\bar{0}}^+))^*+
(\zeta(\mathfrak{g}_{\bar{0}}^-))^*$,  may be  put to be equal to
 a constant, i.e.\ correct reduction is: $ L^{(0)}= K_++
K_-$, where $K_{\pm}$ is a~constant element of
$\zeta(\mathfrak{g}_{\bar{0}}^{\pm})^*$.
\end{proof}

Hence, in this case we have the following form of simplest $U$-$V$
pair  (\ref{uvpts}) for the zero-curvature condition:
\begin{equation}\label{uvptsr}
\nabla I_2^{\pm 2}(u)= u^{\pm 2}K_{\pm} +
  u^{\pm 1} L^{(\pm 1)} + L^{(\pm 2)}_{\pm},
\end{equation}
where $L^{(\pm 2)}_{\pm}=P_0^{\pm}(L^{(\pm 2)})\in
\mathfrak{g}_{\bar{0}}^{\pm}$. The corresponding equations
(\ref{gtes1}), (\ref{gtes2}) are satisf\/ied automatically and the
rest of equations  of the system (\ref{gtes}) are:
\begin{subequations}\label{gtesr}
\begin{gather}\label{gtesr3}
\partial_{x_-} L^{(1)}=[L_-^{(-2)}, L^{(1)}] + [L^{(-1)}, K_+],
\\
\label{gtesr4}
\partial_{x_+} L^{(-1)}=[ L_+^{(2)}, L^{(-1)}]
+[L^{(1)},K_-],
\\
\label{gtesr5}
\partial_{x_+} L_-^{(-2)}- \partial_{x_-} L_+^{(2)}=
 [L^{(1)}, L^{(-1)}],
\end{gather}
\end{subequations}
where we have again used that in our case $[r_{+}(X), r_{-}
(Y)]=0$, $\forall\, X,Y \in \mathfrak{g}_{\bar{0}}$.

The equations (\ref{gtesr3}), (\ref{gtesr4}) are {\it the simplest
possible generalizations of the Thirring equations}. The f\/irst
term in the right-hand-side of this equation is an analog of the
cubic non-linearity of Thirring equation. The second term is an
analog of linear ``massive'' term of the Thirring equations.  In
all the cases one can express $L_{\pm}^{(\pm 2)}$ via  $L^{(\pm
1)}$, and equation (\ref{gtesr5})  will  follow from  equations
(\ref{gtesr3}), (\ref{gtesr4}). In order to show this we will
consider the following example.

\begin{example} Let us consider the case $\mathfrak{g}=gl(n)$, with
the following $Z_2$-grading: $\mathfrak{g}=gl(n)_{\bar{0}} +
gl(n)_{\bar{1}}$, where $gl(n)_{\bar{0}}=gl(p)+ gl(q)$ and
$gl(n)_{\bar{1}}=\mathbb{C}^{2pq}$, i.e. {\samepage
\[
 gl(n)_{\bar{0}}=
\left(
\begin{array}{cc}
  \hat{\alpha} & 0 \\
  0 & \hat{\delta}
\end{array}
\right), \qquad gl(n)_{\bar{1}}=\left(
\begin{array}{cc}
  0 & \hat{\beta} \\
 \hat{\gamma} & 0
\end{array}
\right),
\] where $\hat{\alpha}\in gl(p)$, $\hat{\delta} \in gl(q)$,
$\hat{\beta} \in {\rm Mat}(p,q)$, $\hat{\gamma} \in {\rm
Mat}(q,p)$.}

In this case $\mathfrak{g}_{\bar{0}}^+=gl(p)$,
$\mathfrak{g}_{\bar{0}}^-=gl(q)$,
$\zeta(\mathfrak{g}_{\bar{0}}^+)=k_+1_p$,
$\zeta(\mathfrak{g}_{\bar{0}}^-)=k_-1_q$.  The corresponding
$U$-$V$ pair  (\ref{uvptsr}) has the form:
\begin{gather*}%\label{uvptsr1}
\nabla I_2^{+ 2}(u)= k_+u^{2} \left(%
\begin{array}{cc}
  1_p & 0 \\
  0 & 0
\end{array}%
\right) +
  u \left(%
\begin{array}{cc}
  0 & \hat{\beta}_+ \\
 \hat{\gamma}_+ & 0
\end{array}%
\right) + \left(%
\begin{array}{cc}
  \hat{\alpha}_+ & 0 \\
  0 & 0
\end{array}%
\right),
\\
%\label{uvptsr2}
\nabla I_2^{- 2}(u)= k_-u^{- 2}\left(%
\begin{array}{cc}
 0 & 0 \\
  0 & 1_q
\end{array}%
\right) +
  u^{- 1} \left(%
\begin{array}{cc}
  0 & \hat{\beta}_- \\
 \hat{\gamma}_- & 0
\end{array}%
\right) + \left(%
\begin{array}{cc}
  0 & 0 \\
  0 & \hat{\delta}_-
\end{array}%
\right).
\end{gather*}
The corresponding  zero-curvature condition yields in this case
the following equations:
\begin{subequations}
\begin{gather} \label{vt1}
\partial_{x_-} \hat{\beta}_+=-(\hat{\beta}_+\hat{\delta}_- +k_+\hat{\beta}_-),\qquad
\partial_{x_-} \hat{\gamma}_+=(\hat{\delta}_-\hat{\gamma}_+ +k_+\hat{\gamma}_-),
\\ \label{vt2}
\partial_{x_+} \hat{\beta}_-=(\hat{\alpha}_+\hat{\beta}_- +k_-\hat{\beta}_+),\qquad
\partial_{x_+} \hat{\gamma}_-=-(\hat{\gamma}_-\hat{\alpha}_+ +k_-\hat{\gamma}_+),
\\
\label{vt3}
\partial_{x_+} \hat{\delta}_-=(\hat{\gamma}_+\hat{\beta}_- -\hat{\gamma}_-\hat{\beta}_+),\qquad
\partial_{x_-} \hat{\alpha}_+=(\hat{\beta}_-\hat{\gamma}_+ - \hat{\gamma}_-\hat{\beta}_+).
\end{gather}
\end{subequations}
By direct verif\/ication it is easy to show that the substitution
of variables:
\begin{equation*}
\hat{\delta}_-=-\frac{1}{k_-} \hat{\gamma}_-\hat{\beta}_-, \qquad
\hat{\alpha}_+= -\frac{1}{k_+}\hat{\beta}_+\hat{\gamma}_+
\end{equation*}
solves equation (\ref{vt3}) and yields the following non-linear
 dif\/ferential equations:
%\begin{subequations}\label{vtc}
\begin{gather*} %\label{vtc1}
\partial_{x_-} \hat{\beta}_+= \frac{1}{k_-}\hat{\beta}_+ (\hat{\gamma}_-\hat{\beta}_-) -k_+\hat{\beta}_- ,\qquad
\partial_{x_-} \hat{\gamma}_+= -\frac{1}{k_-}(\hat{\gamma}_-\hat{\beta}_-)\gamma_+  + k_+\hat{\gamma}_- ,
\\
% \label{vtc2}
\partial_{x_+} \hat{\beta}_-= -\frac{1}{k_+}(\hat{\beta}_+\hat{\gamma}_+)\hat{\beta}_- + k_-\hat{\beta}_+ ,\qquad
\partial_{x_+} \hat{\gamma}_-= \frac{1}{k_+}\hat{\gamma}_-(\hat{\beta}_+\hat{\gamma}_+) - k_-\hat{\gamma}_+ .
\end{gather*}
These equations are {\it matrix generalization of the complex
Thirring system} \cite{TW}. They admit several  reductions
(restrictions to the dif\/ferent real form of the algebra
$gl(n)$). For the case of an $u(n)$ reduction we have
\[
\hat{\gamma}_{\pm}=- \hat{\beta}^{\dag}_{\pm}= -
\Psi^{\dag}_{\pm}, \qquad k_{\pm}=i\kappa_{\pm},\qquad
\kappa_{\pm} \in \mathbb{R}
\] and we
obtain the following non-linear matrix equations in partial
derivatives:
%\begin{subequations}\label{vtr}
\begin{gather*}% \label{vtr1}
i\partial_{x_-} \Psi_+= \frac{1}{\kappa_-}\Psi_+
(\Psi^{\dag}_{-}\Psi_{-}) +\kappa_+\Psi_{-} ,
\qquad %\label{vtr2}
i\partial_{x_+} \Psi_-=-
\frac{1}{\kappa_+}(\Psi_{+}\Psi^{\dag}_{+})\Psi_-+ \kappa_-\Psi_+.
\end{gather*}
These equations are {\it a matrix generalization of the massive
Thirring equations} ( $\Psi_{\pm} \in {\rm Mat}(p,q)$). In the
case of
 $p=n-1$, $q=1$ or $p=1$, $q=n-1$ we obtain {\it a
vector generalization of Thirring equations}. In the special case
$n=2$, $p=q=1$ and $\kappa_-= \kappa_+$ we recover the usual
scalar massive Thirring equations with mass $m=(\kappa_+)^2$.
\end{example}

\subsubsection{Case of Coxeter automorphism}

Let $\sigma$ be a  Coxeter automorphism. In this case $p=h$, where
$h$ is a Coxeter number of $\mathfrak{g}$ and algebra
$\mathfrak{g}_{\overline{0}}$ is Abelian. That is why all tensors
$R_0$
 are solutions of the mYBE on
$\mathfrak{g}_{\overline{0}}$ and maps~$r_{\pm}$ are arbitrary
(modulo the constraint $r_{+}+ r_{-}=1$). Moreover, in this case
 the following proposition holds:

\begin{proposition}
Let $\sigma$ be a Coxeter automorphism of $\mathfrak{g}$ and maps
$r_{\pm}$ on $\mathfrak{g}$ be nondegenerated. In this case
$L^{(0)}$ is constant along all time flows and  components of
$L^{(\pm h)}$ are expressed  as polynomials of  components of
$L^{(\pm 1)},\dots, L^{(\pm (h-1))}$.
\end{proposition}

\begin{proof}
 In the considered case we have that $\mathrm{dim}\,
\mathfrak{g}_{\overline{0}}=\mathrm{rank}\, \mathfrak{g}$. On the
other hand there exist $r=\mathrm{rank}\, \mathfrak{g}$
independent Casimir functions $I_k(L)$, $k\in \overline{1,h}$, on
$\mathfrak{g}^*$.  Hence there exist $2\, \mathrm{rank}\,
\mathfrak{g}$ integrals of the following form:
\begin{gather*}
 I_k^{\pm 0}= I_k( r_{\pm}^*( L^{(0)})).
\end{gather*}
They are constant along all time f\/lows generated by all other
integrals $I_k^{\pm l}$. That is why we may put them to be equal
to constants. On the other hand, it is not dif\/f\/icult to see
that all $r$ independent components of $L^{(0)}$ can be
functionally expressed via $I_k^{\pm 0}$. Hence, they are
constants too. At last, we have $2\, \mathrm{rank}\, \mathfrak{g}$
integrals $I_k^{\pm h}$, which are constant along all time f\/lows
and we may put them to be equal to constants $I_k^{\pm h}={\rm
const}_k^{\pm h}$. It easy to see that $I_k^{\pm h}$ are linear in
$L^{(\pm h)}$ and, on the surface of level of $I_k^{\pm h}$,  all
components of $L^{(\pm h)}$ are expressed polynomially via
components of $L^{(\pm 1)},\dots,  L^{(\pm (h-1))}$ if the maps
$r_{\pm}$ are non-degenerate.
%Proposition is proved.
\end{proof}
This proposition has the following important corollary:
\begin{corollary}
The number of independent fields in the generalized Thirring
equation \eqref{zce2t}, corresponding to Coxeter automorphism, is
equal to  $2 \sum\limits_{j=1}^{h-1} \dim \mathfrak{g}_j =2 (\dim
\mathfrak{g}- \mathrm{rank}\, \mathfrak{g})$.
\end{corollary}

Let us explicitly consider Thirring-type equations (\ref{gte}) in
the case of the Coxeter automorphisms. In this case
$\mathfrak{g}_{\overline{0}}\simeq  \mathfrak{g}_{\overline{h}}$
is Abelian and equations (\ref{gte1}), (\ref{gte2}) become
trivial. Moreover, due to the fact that $L^{(\pm h)}$ is expressed
via $L^{(\pm j)}$ where $j<h$,  equation (\ref{gte4}) becomes a
consequence of  equations (\ref{gte3}), (\ref{gte4}). In the
resulting system of equations (\ref{gte}) is simplif\/ied to the
following  system:
%\begin{subequations}
\begin{gather}%\label{gtec1}
\partial_{x_-} L^{(h-j)}=[ r_-(L^{(-h)}), L^{(h-j)})] +
\sum\limits_{i=j+1}^{h-1} [ L^{(i-j-h)}, L^{(h-i)}] +[L^{(-j)},
r^*_+(L^{(0)})],
\nonumber\\
%\label{gtec2}
\partial_{x_+} L^{(j-h)}=[ r_+(L^{(h)}), L^{(j-h)})] +
\sum\limits_{i=j+1}^{h-1} [L^{(h+j-i)}, L^{(j-h)}]
+[L^{(j)},r^*_-(L^{(0)})],\label{gtec}
\end{gather}
%\end{subequations}
where $j\in \{1,h-1\}$, $L^{(0)}$ is a constant matrix and
$L^{(\pm h)}$ is polynomial in  $L^{(\pm j)}$. We call this system
of equations {\it the Abelian generalized  Thirring equations}.
The last summand in  equations~(\ref{gtec}) is an analog of the
``massive term'' in the Thirring equation. The f\/irst summand is
an analog of the cubic non-linearity in the Thirring equation, but
with the growth of the rank of $\mathfrak{g}$ the degree of this
term is also growing. The other terms are of the second order in
dynamical variables. They are absent in the ordinary Thirring
system corresponding to the case
$\widetilde{\mathfrak{g}}=\widetilde{sl(2)}$. Let us explicitly
consider the simplest   example, which already possesses all
features of the generalized Thirring equation:

\begin{example}\label{ex5} Let $\mathfrak{g}=gl(3)$. Its $Z_3$-grading
corresponding to the  Coxeter automorphism has the following form:
\begin{gather*}
gl(3)_{\bar{0}}=\left(%
\begin{array}{ccc}
  \alpha_1 & 0 & 0 \\
  0 & \alpha_2& 0 \\
  0 & 0 & \alpha_3
\end{array}%
\right), \quad  gl(3)_{\bar{1}}=\left(%
\begin{array}{ccc}
  0 & \beta_1 & 0 \\
  0 & 0 & \beta_2 \\
  \beta_3 & 0 & 0
\end{array}%
\right), \quad gl(3)_{\bar{2}}=\left(%
\begin{array}{ccc}
  0 & 0 & \gamma_3 \\
  \gamma_1 & 0 & 0 \\
  0 & \gamma_2 & 0
\end{array}%
\right),
\end{gather*} $gl(3)_{\bar{3}}=gl(3)_{\bar{0}}$,
$gl(3)_{\overline{-1}}=gl(3)_{\bar{2}}$,
$gl(3)_{\overline{-2}}=gl(3)_{\bar{1}}$.

The Lax operators belonging to the dual spaces to
$\widetilde{gl(3)}_{R_{\pm}}$ are:{\samepage
\begin{equation*}
L^{\pm}(u)= r^*_{\pm}( L^{(0)}) + L^{(\pm 1)} u^{\mp 1} + L^{(\pm
2)} u^{\mp 2} +  L^{(\pm 3)} u^{\mp 3}+\cdots,
\end{equation*}
where $L^{(0)}, L^{(\pm 3)} \in gl(3)_{\bar{0}}$, $L^{(1)},
L^{(-2)} \in gl(3)_{\bar{2}}$, $L^{(2)}, L^{(-1)} \in
gl(3)_{\bar{1}}$.}

 In order to simplify the form of the resulting soliton equations we
 will use the following notations:
\begin{gather*}
L^{(0)}=\left(%
\begin{array}{ccc}
  \alpha_1 & 0 & 0 \\
  0 & \alpha_2& 0 \\
  0 & 0 & \alpha_3
\end{array}%
\right), \quad  L^{(2)}=\left(%
\begin{array}{ccc}
  0 & \gamma^+_1 & 0 \\
  0 & 0 & \gamma^+_2 \\
  \gamma^+_3 & 0 & 0
\end{array}%
\right),\quad  L^{(1)}=\left(%
\begin{array}{ccc}
  0 & 0 & \beta^+_3 \\
  \beta^+_1 & 0 & 0 \\
  0 & \beta^+_2 & 0
\end{array}%
\right),\\
 L^{(\pm 3)}=\left(%
\begin{array}{ccc}
  \delta^{\pm }_1 & 0 & 0 \\
  0 & \delta^{\pm}_2& 0 \\
  0 & 0 & \delta^{\pm}_3
\end{array}%
\right), \quad  L^{(-1)}=\left(%
\begin{array}{ccc}
  0 & \beta^-_1 & 0 \\
  0 & 0 & \beta^-_2 \\
  \beta^-_3 & 0 & 0
\end{array}%
\right),\quad L^{(-2)}=\left(%
\begin{array}{ccc}
  0 & 0 & \gamma^-_3 \\
  \gamma^-_1 & 0 & 0 \\
  0 & \gamma^-_2 & 0
\end{array}%
\right).
\end{gather*}

 The
generating functions of the Poisson-commuting integrals of the
corresponding integrable system are:
\[ I^{k}(L^{\pm }(u))=\frac{1}{k} \,{\rm tr}\,
(L^{\pm}(u))^k=\sum\limits_{m=0}^{\infty} I_{k}^{\pm m}u^{\mp m}.
\]
In particular, we have the following integrals:
\[
I_{k}^{\pm 0}=\frac{1}{k} \,{\rm tr}\, (r^*_{\pm}(L^{(0)}))^k ,
\qquad k\in \overline{1,3}.
\]
From the fact that they are f\/ixed along all the time f\/lows we
obtain that all $\alpha_i$ are constants of motion. We also have
the following integrals:
\begin{gather*}
I_{1}^{\pm 3}= {\rm tr}\, L^{(\pm 3)},\qquad
 I_{2}^{\pm 3}={\rm tr}\, (r^*_{\pm}(L^{(0)}) L^{(\pm 3)}) + {\rm tr}\, (L^{(\pm 1)} L^{(\pm 2)}),\\
I_{3}^{\pm 3}= \frac{1}{3}\,{\rm tr}\, (L^{(\pm 1)})^{3} + {\rm
tr}\,\bigl(r^*_{\pm}(L^{(0)}) (L^{(\pm 1)} L^{(\pm 2)} + L^{(\pm
2)} L^{(\pm 1)})\bigr) + {\rm tr}\,\bigl(r^*_{\pm}(L^{(0)})^2
L^{(\pm 3)}\big).
\end{gather*}
These integrals permit us to express components of $L^{(\pm 3)}$
via components of $L^{(\pm 1)}$ and  $L^{(\pm 2)}$ in the case of
the non-degenerated maps $r_{\pm}$. Let us do this explicitly. For
the sake of  simplicity we will put that
$r^*_{\pm}(L^{(0)})=k_{\pm} L^{(0)}$, where $k_{\pm}$ are some
constants and $k_+ + k_-=1$. In this case we will have the
following explicit form of the above integrals:{\samepage
\begin{gather*}
I_{1}^{\pm 3}= \sum\limits_{i=1}^3 \delta^{\pm}_i,\qquad
I_{2}^{\pm 3}= k_{\pm}
\sum\limits_{i=1}^3 \alpha_i\delta^{\pm}_i + \sum\limits_{i=1}^3 \beta^{\pm}_i\gamma^{\pm}_i,\\
I_{3}^{\pm 3}= k^2_{\pm} \sum\limits_{i=1}^3
\alpha^2_i\delta^{\pm}_i + k_{\pm} \sum\limits_{i=1}^3
\alpha_i(\beta^{\pm}_i\gamma^{\pm}_i +
\beta^{\pm}_{i-1}\gamma^{\pm}_{i-1}) + \beta^{\pm}_1 \beta^{\pm}_2
\beta^{\pm}_3,
\end{gather*}
where we have implied  in the last summation  that
$\beta^{\pm}_{0}\equiv \beta^{\pm}_{3}, \gamma^{\pm}_{0} \equiv
\gamma^{\pm}_{3}$.}

The $U$-$V$ pair corresponding to the Hamiltonians $I_{2}^{\pm 3}$
have the form:
\begin{equation*}
\nabla I_2^{\pm 3}(u)= u^{\pm 3} r^*_{\pm}( L^{(0)}) +
  u^{\pm 2} L^{(\pm 1)} + u^{\pm 1} L^{(\pm 2)}+ r_{\pm} ( L^{(\pm 3)}).
\end{equation*}
The corresponding zero-curvature equation reads as follows:
%\begin{subequations}\label{gtes'}
\begin{gather*}%\label{gtes1'}
\partial_{x_-} L^{(1)}=[r_-(L^{(-3)}), L^{(1)}] + [L^{(-2)}, r^*_+(L^{(0)})],
\\
%\label{gtes3'}
\partial_{x_+} L^{(-1)}=[ r_+(L^{(3)}), L^{(-1)}]
+[L^{(2)},r^*_-(L^{(0)})],
\\
%\label{gtes2'}
\partial_{x_-} L^{(2)}=[r_-(L^{(-3)}), L^{(2)}] + [L^{(-1)}, r^*_+(L^{(0)})] + [L^{(1)},
 L^{(-2)}],
\\
%\label{gtes4'}
\partial_{x_+} L^{(-2)}=[ r_+(L^{(3)}), L^{(-2)}]
+[L^{(1)},r^*_-(L^{(0)})] +  [L^{(2)},
 L^{(-1)}].
\end{gather*}
In the component form  we  have the following equations:
\begin{gather}%\label{gtes1''}
\partial_{x_-} \beta_i^+ =k_- (\delta^{-}_{i}- \delta^{-}_{i+1})\beta_i^+ +
k_+ (c_{i+1}- c_{i})\gamma_i^-,
\nonumber\\
%\label{gtes3''}
\partial_{x_+} \beta_i^- =k_+ (\delta^{+}_{i}- \delta^{+}_{i+1})\beta_i^- +
k_- (c_{i+1}- c_{i})\gamma_i^+,
\nonumber\\
%\label{gtes2''}
\partial_{x_-} \gamma_i^+ =k_- (\delta^{-}_{i+1}- \delta^{-}_{i})\gamma_i^+ +
k_+ (c_{i}- c_{i+1})\beta_i^- +( \beta_k^ + \gamma_j^ - - \beta_j^
+ \gamma_k^ -),
\nonumber\\
%\label{gtes4''}
\partial_{x_+} \gamma_i^+ =k_+ (\delta^{-}_{i+1}- \delta^{-}_{i})\gamma_i^+ +
k_- (c_{i}- c_{i+1})\beta_i^- +(  \gamma_j^ + \beta_k^ - -
\gamma_k^ +\beta_j^ - ),\label{gtes''}
\end{gather}
where $c_i \equiv \alpha_i$, indices $i$, $j$, $k$ constitute a
cyclic permutation of the indices $1$, $2$, $3$, and it is implied
that $\delta^{\pm }_{3+1}\equiv \delta^{\pm}_{1}$, $c_{3+1}\equiv
c_{1}$.

Taking into account that all the integrals $I^{k}_{\pm 3}$ are
constants of motion and  putting their values to be equal to zero
we can  explicitly express variables $\delta^{\pm}_i$ via
$\alpha_i=c_i$, $\beta^{\pm}_i$, $\gamma^{\pm}_i$:
\begin{gather}
\delta^{\pm}_i= \dfrac{1 }{k_{\pm}^2(c_i-c_j)(c_i -c_k)}\nonumber\\
\phantom{\delta^{\pm}_i=}{}\times\left(k_{\pm} (c_j+c_k)
\sum\limits_{l=1}^3 \beta^{\pm}_l\gamma^{\pm}_l -  k_{\pm}
\sum\limits_{l=1}^3 c_l(\beta^{\pm}_l\gamma^{\pm}_l +
\beta^{\pm}_{l-1}\gamma^{\pm}_{l-1}) - \beta^{\pm}_1 \beta^{\pm}_2
\beta^{\pm}_3\right),\label{expr}
\end{gather}
where indices $j$ and $k$ are complementary to the index $i$ in
the set $\{1,2,3\}$. Due to the fact that $c_i$ are constants of
motion this expression is  polynomial of the third order in
dynamical variables.

At last, substituting (\ref{expr}) into the equations
(\ref{gtes''}) we obtain a system of nonlinear dif\/ferential
equations for the dynamical variables  $\beta^{\pm}_i$ and
$\gamma^{\pm}_i$. These equations are a $gl(3)$ generalization of
``complex'' Thirring equations. Unfortunately, neither these
equations nor  more complicated~$gl(n)$ complex Thirring equations
do not admit reductions to the real forms $u(3)$ or $u(n)$
respectively.
\end{example}

\subsection[The case of Kostant-Adler-Symes $R$-operators]{The case of Kostant--Adler--Symes $\boldsymbol{R}$-operators}
Let us now consider integrable hierarchies corresponding to the
case $r_+=1$, $r_-=0$ (or $r_-=0$, $r_-=1$), i.e.\ corresponding
to the Kostant--Adler--Symes decomposition of t loop algebras.

As it was reminded in  Subsection~\ref{loop}, each of the gradings
of the loop algebras corresponding to the dif\/ferent
automorphisms $\sigma$ of order $p$ provides  its own
Kostant--Adler--Symes decomposition:
$\widetilde{\mathfrak{g}}^{\sigma}=\widetilde{\mathfrak{g}}^{\sigma+}+
\widetilde{\mathfrak{g}}^{\sigma-}$, and, hence, provides
commuting Hamiltonian f\/lows and hierarchies of integrable
equations.

 Let us
consider the simplest equations of integrable hierarchies
corresponding to  them. We have the following generic elements of
the dual spaces $(\widetilde{\mathfrak{g}}^{\sigma \pm})^*$:
\[
L_{+}(u)= L^{(0)} + \sum\limits_{j=1}^{ \infty } L^{(j)}
u^{-j},\qquad L_{-}(u)=  \sum\limits_{j=1}^{ \infty } L^{(-j)}
u^{j},
\]
where $L^{(-k)}\in \mathfrak{g}_{\bar{k}}$, and the following
generating functions of the commutative integrals:
\begin{equation*}
I_k(L^+(u))= \sum\limits_{m=0}^{\infty} I^m_k u^{-m},\qquad
I_l(L^-(u))= \sum\limits_{n=1}^{\infty} I^{-n}_l u^{n},
\end{equation*}
where $I_k(L)$ are Casimir functions of $\mathfrak{g}$.

We will be interested in the following Hamiltonians and their
Hamiltonian f\/lows:
\begin{gather*}
  I_2^{p}= \bigl(L^{(0)}, L^{(p)} \bigr) +
   \frac{1}{2} \sum\limits_{j=1}^{p-1}
\bigl(  L^{(j)}, L^{(p-j)} \bigr),\qquad
 I_p^{-p}=I_p( L^{(-1)}).
\end{gather*}
The algebra-valued gradients of functions $I_2^{p}$, $I_p^{-p}$
read as follows:{\samepage
\begin{equation*}
\nabla I_2^{p}(u)= u^{p}  L^{(0)} +
  \sum\limits_{j=1}^{p-1} u^{(p-j)}
  L^{(j)} +  L^{( p)},\qquad
\nabla I_p^{-p}(u)= u^{-1} \tilde{L}^{(1)},
\end{equation*}
where $\tilde{L}^{(1)}\equiv \sum\limits_{i=1}^{ \dim
\mathfrak{g}_{\bar{1}}}\frac{\partial I_p^{-p} }{\partial
L^{(1)}_{\alpha} } X^{\overline{-1}}_{\alpha} \in
\mathfrak{g}_{\overline{-1}} $.}

The  corresponding zero-curvature condition yields the following
system of  equations:
\begin{gather}\label{tte}
\partial_{x_-} L^{(0)}= 0,\\ \label{seq}
\partial_{x_-} L^{(k)}= [\tilde{L}^{(1)},L^{(k-1)}], \qquad k \in \{1,p\},\\  \label{teq}
\partial_{x_+} \tilde{L}^{(1)}= [L^{(p)}, \tilde{L}^{(1)}].
\end{gather}

\begin{remark} The gradient $\nabla I_2^{p}(u)$ is an analog of
the $U$ operator of the generalized Thirring hierarchy. The other
gradient $\nabla I_p^{-p}(u)$ is an analog of the $V$ operator of
the non-Abelian Toda equation. Hence the corresponding integrable
equations  may be considered as an intermediate case between
generalized Thirring and (non-Abelian) Toda equation.
\end{remark}

\subsubsection{Case of  second order automorphism}

Let us consider the case of the automorphism of the second order
$\sigma^2=1$ ($p=2$), the correspon\-ding Hamiltonians, their
matrix gradients and zero-curvature conditions. We will use
commuting second order Hamiltonians of  the following form:
\begin{gather*}
  I_2^{2}= \bigl(L^{(0)}, L^{(2)} \bigr) +
\tfrac{1}{2}\bigl(  L^{(1)}, L^{(1)} \bigr),\qquad
 I_2^{-2}= \tfrac{1}{2} (L^{(-1)}, L^{(-1)}).
\end{gather*}
The algebra-valued gradients of functions $I_2^{2}$, $I_2^{-2}$
read as follows:
\begin{equation*}
\nabla I_2^{2}(u)= u^{2}  L^{(0)} +u L^{(1)} +  L^{( 2)},\qquad
\nabla I_p^{-p}(u)= u^{-1} {L}^{(-1)}.
\end{equation*}
The corresponding  zero-curvature condition yields the following
simple system of dif\/ferential equations of hyperbolic type:
\begin{gather}
\partial_{x_-} L^{(0)}= 0,\qquad
\partial_{x_-} L^{(1)}= [{L}^{(-1)},L^{(0)}],\qquad
\partial_{x_-} L^{(2)}= [{L}^{(-1)},L^{(1)}],\nonumber\\
\partial_{x_+} {L}^{(-1)}= [L^{(2)},{L}^{(-1)}],\label{soa}
\end{gather}
where $L^{(0)}, L^{(2)} \in \mathfrak{g}_{\bar{0}}$, $L^{(\pm
1)}\in \mathfrak{g}_{\bar{1}}$.

Let us consider the following example of these equations:

\begin{example} Let $\mathfrak{g}=sl(2)$ and $\sigma$ be a Cartan
involution, i.e.\ $sl(2)_{\bar{0}}={\rm diag}\,(\alpha,-\alpha)$
and $sl(2)_{\bar{1}}$ consists of the matrices with  zeros on the
diagonal. The simplest  Hamiltonians obtained in the framework of
Kostant--Adler--Symes scheme are:
\begin{gather*}
 I_2^{ 0}= (\alpha^{(0)})^2,\qquad
 I_2^{ 2}= 2\alpha^{(0)}\alpha^{( 2)} + \beta^{( 1)}\gamma^{( 1)},\qquad
 I_2^{2}=  \beta^{(-1)}\gamma^{(-1)}.
\end{gather*}
They produce the following $U$-$V$ pair for zero-curvature
equations:
\begin{gather*}
 \nabla I_2^{2}= \left(%
\begin{array}{cc}
  \alpha^{(2)} & 0 \\
  0 & -\alpha^{( 2)}
\end{array}%
\right) + u\left(%
\begin{array}{cc}
   0 & \beta^{(1)} \\
  \gamma^{(1)} & 0
\end{array}%
\right)+u^{2}\left(%
\begin{array}{cc}
  \alpha^{(0)} & 0 \\
  0 & -\alpha^{(0)}
\end{array}%
\right),\\
\nabla I_2^{-2}=  u^{-1}\left(%
\begin{array}{cc}
   0 & \beta^{(-1)} \\
  \gamma^{(-1)} & 0
\end{array}%
\right).
\end{gather*}
Due to the fact that $I_2^{0}$ are constants of motion we have
that $\alpha^{(0)}$ is also a constant of motion. Moreover, using
the fact that $I_2^{2}$ is a constant of motion one can express
$\alpha^{(2)}$ via~$\beta^{(1)}$,~$\gamma^{(1)}$. Hence we obtain
from the equations (\ref{soa}) the following independent equations
for the variab\-les~$\beta^{(\pm 1)}$,~$\gamma^{(\pm 1)}$:
\begin{subequations}\label{comthir'}
\begin{gather}\label{comthir1'}
\partial_{x_-}\gamma^{(1)}=2\alpha^{(0)}\gamma^{(-1)} ,
\\
\label{comthir2'}
\partial_{x_-}\beta^{(1)}=-2\alpha^{(0)}\beta^{(-1)} ,
\\
\partial_{x_+}\gamma^{(-1)}=-2\alpha^{(2)} \gamma^{(-1)},
\\
\label{comthir4'}
\partial_{x_+}\beta^{(-1)}= 2\alpha^{(2)} \beta^{(-1)}.
\end{gather}
\end{subequations}
From the equation (\ref{comthir2'})  it follows that
$\beta^{(-1)}=-\frac{\partial_{x_-}\beta^{(1)}}{2\alpha^{(0)}}$.
Substituting this into equation (\ref{comthir4'}) and taking into
account that $\alpha^{(0)}$  is a constant along all time f\/lows
we obtain the following equation:
\begin{equation}\label{prom}
\partial^2_{x_+x_-}\beta^{(1)}= 2\alpha^{(2)} \partial_{x_-}\beta^{(1)}.
\end{equation}
Taking into account that $\alpha^{( 2)}=-\frac{1}{2\alpha^{(0)}}
(\beta^{( 1)}\gamma^{( 1)})$, where we have put that $I^2_2=0$,
and making the reduction to the Lie algebra $su(2)$:
$\alpha^{(0)}=ic$, $\gamma^{(1)}= -\bar{\beta}^{( 1)}\equiv
-\bar{\psi} $ we obtain the following integrable equation in
partial derivatives:
\begin{equation*}
\partial^2_{x_+x_-}\psi+ \frac{i}{c} |\psi|^2\partial_{x_-}\psi=0.
\end{equation*}
This equation is (in a some sense) intermediate between Thirring
 and sine-Gordon equations.
\end{example}

\subsubsection{Case of Coxeter automorphism (principal gradation)}

Let us  consider again the case of the principal gradation. In
this case $p=h$ and a subalgebra~$\mathfrak{g}_{\bar{0}}$ is
Abelian. In the same way as it was done in the case of the
``generalized Abelian Thirring models''  it is possible to show
that~$L^{(0)}$ is constant along all time f\/lows and components
of~$L^{(h)}$ are expressed polynomially via the components
of~$L^{(k)}$, $k<h$. Let us assume that constants of motion
$L^{(0)}$  are such that the operator ${\rm ad}_{L^{(0)}}$ is
nondegenerate. In such a case we may solve the f\/irst of the
equations~(\ref{seq}) in the following way:
\begin{equation*}
\tilde{L}^{(1)}=-{\rm ad}^{-1}_{L^{(0)}}(\partial_{x_-} L^{(1)}).
\end{equation*}
Substituting this expression into  equation (\ref{teq}) and taking
into account commutativity of $\mathfrak{g}_{\bar{0}}$ and, hence,
operators ${\rm ad}^{-1}_{L^{(0)}}$ and  ${\rm ad}_{L^{(h)}}$, we
f\/inally obtain the following matrix dif\/ferential equation in
partial derivatives:
\begin{equation*}
\partial_{x_+x_-}{L}^{(1)}=
\big[L^{(h)}({L}^{(1)},\dots,{L}^{(h-1)}),\partial_{x_-}{L}^{(1)}\big],
\end{equation*}
where ${L}^{(k)}$, $k \in \{2,h-1\}$ satisfy the following set of
ordinary dif\/ferential equations:
\begin{equation*}
\partial_{x_-} L^{(k)}= [L^{(k-1)},{\rm ad}^{-1}_{L^{(0)}}(\partial_{x_-} L^{(1)})].
\end{equation*}

Let us consider the following example.

\begin{example} Let $\mathfrak{g}=gl(3)$, $h=3$. In this case we
have the following dif\/ferential equations:
\begin{gather}\label{tte31}
\partial_{x_+x_-}{L}^{(1)}=
[L^{(3)}({L}^{(1)},{L}^{(2)}),\partial_{x_-}{L}^{(1)}],
\\
% \label{tte32}
\partial_{x_-} L^{(2)}= [L^{(1)},ad^{-1}_{L^{(0)}}(\partial_{x_-}
L^{(1)})].\nonumber
\end{gather}
Let  $L^{(0)}$, $ L^{(1)}$, $L^{(2)}$, $L^{(3)}$ are parametrized
as in the  Example~\ref{ex5}, i.e.
\begin{gather*}
L^{(0)}=\left(%
\begin{array}{ccc}
  \alpha_1 & 0 & 0 \\
  0 & \alpha_2& 0 \\
  0 & 0 & \alpha_3
\end{array}%
\right),\qquad
 L^{(1)}=\left(%
\begin{array}{ccc}
  0 & 0 & \beta_3 \\
  \beta_1 & 0 & 0 \\
  0 & \beta_2 & 0
\end{array}%
\right), \\
L^{(2)}=\left(%
\begin{array}{ccc}
  0 & \gamma_1 & 0 \\
  0 & 0 & \gamma_2 \\
  \gamma_3 & 0 & 0
\end{array}%
\right),\qquad
L^{(3)}=\left(%
\begin{array}{ccc}
  \delta_1 & 0 & 0 \\
  0 & \delta_2& 0 \\
  0 & 0 & \delta_3
\end{array}%
\right).
\end{gather*}
 In such coordinates equation (\ref{tte31}) acquires the
following form:
\begin{equation} \label{tte31'}
\partial^2_{x_- x_+} \beta_i= (\delta_{i+1}(\beta,\gamma)- \delta_{i}(\beta,\gamma))\partial_{x_-}
\beta_i, \qquad i\in \{1,3\}.
\end{equation}
where, like in the Example \ref{ex5}, $\delta_{i}(\beta,\gamma)$
are expressed via $\beta_j$, $\gamma_k$ and constants
$\alpha_i\equiv c_i$:
\begin{gather}\label{expr2}
\delta_i(\beta,\gamma)= \dfrac{1 }{(c_i-c_j)(c_i -c_k)}\left(
(c_j+c_k) \sum\limits_{l=1}^3 \beta_l\gamma_l -
\sum\limits_{l=1}^3 c_l(\beta_l\gamma_l + \beta_{l-1}\gamma_{l-1})
- \beta_1 \beta_2 \beta_3\right).
\end{gather}
This equation is an exact analog of the equation (\ref{prom}).
Unfortunately in this case there is no $su(3)$ reduction and
variables $\gamma_i$ are not conjugated to $\beta_i$ but satisfy
the following dif\/ferential equations:
\begin{equation} \label{tte32'}
\partial_{x_- } \gamma_i= \epsilon_{ijk}\beta_j(c_k-c_{k+1})^{-1}\partial_{x_- }
\beta_k,  \qquad i\in \{1,3\}.
\end{equation}
Equations  (\ref{tte31'})--(\ref{tte32'}) are intermediate between
the Abelian $gl(3)$-Toda equations and the genera\-li\-zed Abelian
$gl(3)$-Thirring equations.
\end{example}

\subsection*{Acknowledgements} Author is grateful to Professors A.~Mikhailov and  P.~Holod
for attracting his attention to Thirring model and to M.~Tsuchida
for useful discussion.

\pdfbookmark[1]{References}{ref}
\LastPageEnding
\end{document}